\documentclass[pra,showpacs,aps,twocolumn]{revtex4}
\usepackage{amsmath}
\usepackage{graphicx}
\usepackage{dcolumn}
\usepackage{amssymb}
\usepackage{color}

\usepackage{graphicx,comment}% Include figure files
\usepackage{epstopdf}
\usepackage{bm}% bold math
\usepackage{enumerate}
\usepackage{fancyhdr,color}
\usepackage{verbatim}
\usepackage{amsmath,amsthm}
\usepackage{amsfonts}
\usepackage{hyperref}

\def\cM{\mathcal{M}}

\def\cS{\mathcal{S}}

\def\cU{\mathcal{U}}

\def\spi{{\sqrt{\pi}}}

\def\Tr{\mathrm{Tr}}
\def\Pr{\mathrm{Pr}}

\newtheorem{definition}{Definition}
\newtheorem{theorem}{Theorem}
\newtheorem{lemma}{Lemma}
\newtheorem{corollary}{Corollary}

\def\one{{\mathchoice {\rm 1\mskip-4mu l} {\rm 1\mskip-4mu l} {\rm
1\mskip-4.5mu l} {\rm 1\mskip-5mu l}}}
\newcommand{\ket}[1]{\left| #1\right\rangle}        % ket vector
\newcommand{\sket}[1]{| #1\rangle}  
  
\newcommand{\bra}[1]{\left\langle #1\right|}        % bra vector
\begin{document}

%\preprint{APS/123-QED}

\title{Quantum algorithms for Gibbs sampling and hitting-time estimation}

\author{
Anirban Narayan Chowdhury}
\affiliation{Center for Quantum Information and  Control, University of New Mexico, Albuquerque, NM 87131, US and New Mexico Consortium,
Los Alamos,
NM 87545, US.}

\author{
Rolando D. Somma}
\affiliation{Theoretical Division, Los Alamos National Laboratory, Los Alamos, NM 87545, US.}

\date{\today}

%%%%%%%%%%%%%%%%%%%%%%%%%%%%%%%%%%%%%%%%%%%%%%%%%%
  \begin{abstract}
  We present quantum algorithms for solving two problems regarding stochastic processes.
  The first algorithm prepares the thermal Gibbs state of a quantum system
  and runs in time almost linear in $\sqrt{N \beta/{\cal Z}}$ and polynomial in $\log(1/\epsilon)$, 
  where $N$ is the Hilbert space dimension, $\beta$
  is the inverse temperature,  ${\cal Z}$ is the partition function, and $\epsilon$ is the desired
  precision of the output state. Our quantum algorithm exponentially improves the dependence
  on $1/\epsilon$ and quadratically improves the dependence on $\beta$ of known quantum algorithms
  for this problem. The second algorithm estimates the hitting time of a Markov chain. For a sparse stochastic matrix $P$, it runs in time
  almost linear in $1/(\epsilon \Delta^{3/2})$, where $\epsilon$ is the absolute precision in the estimation
  and $\Delta$ is a parameter determined by $P$, and whose inverse is an upper bound of the hitting time. 
  Our quantum algorithm quadratically improves the dependence on $1/\epsilon$ and $1/\Delta$
  of the analog classical algorithm for hitting-time estimation. Both algorithms use  tools recently developed
  in the context of Hamiltonian simulation, spectral gap amplification, and solving linear systems of equations.
  \end{abstract}
  
  
  \maketitle
  
 %%%%%%%%%%%%%%%%%%%%%%%%%%%%%%%%%%%%%%%%%%%%%%%%%%
 %%%%%%%%%%%%%%%%%%%%%%%%%%%%%%%%%%%%%%%%%%%%%%%%%%

\section{Introduction}
Two important problems in statistical mechanics  and stochastic processes
are sampling from the thermal or Gibbs distribution of a physical system at a certain temperature
and the estimation of hitting times of  classical Markov chains.
The first such problem has a wide range of applications
as it  allows us to compute quantities like the partition function,
energy, or entropy of the system, and understand its physical properties in thermal equilibrium~\cite{Hua63}.
This problem has also applications in many other scientific areas including optimization~\cite{Aga95}.
Hitting times are also paramount in the study of classical random processes
and they allow for a characterization of Markov chains~\cite{LPW}. Roughly, a hitting time
is the time required by a diffusive random walk to reach a particular configuration
with high probability. Besides their use in physics, hitting times are also important
in solving search problems where the goal is to find a marked configuration
of the Markov chain~\cite{RP13}.

In a classical setting, these two problems are commonly solved
using Monte-Carlo techniques~\cite{NB98}. Each step in a Monte-Carlo simulation
corresponds to applying a particular probability rule that determines a Markov chain and an associated
 stochastic matrix. In the case of sampling from Gibbs distributions, for example, the fixed
point of the Markov chain (i.e., the eigenvector of the stochastic matrix with eigenvalue 1) 
corresponds to the desired distribution. Such a distribution can
then  be prepared
by repeated
applications of the probability rule. To sample from probability distributions associated
with thermal Gibbs states of quantum systems, quantum Monte-Carlo
techniques may be used~\cite{And07}.
The running time of a Monte-Carlo simulation is typically dominated
by the number of times the probability rule is applied to prepare the desired distribution
with some given precision.
This running time depends on properties of the Markov chain such as 
the spectral gap of the stochastic matrix~\cite{LPW}.

In recent years, there has been significant interest in the 
development of quantum algorithms for simulating stochastic processes.
Quantum algorithms for thermal Gibbs state preparation were developed
in various works (c.f.,~\cite{SBB08,PW09,CW10,BB10,TOVPV11,STV12}) and showed to 
provide polynomial quantum speedups in terms of various parameters, such as the spectral
gap of the stochastic matrix or the dimension of the Hilbert space.
The notion of quantum hitting time was also introduced in numerous works (c.f.,~\cite{AKR05,Sze04,KB06,MNR07,KOR10}).
 Often, quantum hitting times of quantum walks on different graphs are significantly (e.g., polynomially) smaller than their classical counterparts. 
 There are also various quantum algorithms to accelerate classical Monte-Carlo methods for estimating different quantities,
 such as expected values or partition functions (c.f.,~\cite{KOS07,Mon15}).
 Our results  advance these areas further by providing new quantum algorithms with various improvements in the running time with respect to known
 classical and quantum algorithms for some of these problems.

In more detail, we present two quantum algorithms 
for preparing thermal Gibbs states of quantum systems
 and for estimating hitting times, respectively.
 The first algorithm runs in time $\tilde O(\sqrt{N \beta/{\cal Z}})$,
 where $N$ is the Hilbert space dimension, $\beta$ is the inverse temperature,
 and ${\cal Z}$ is the partition function of the quantum system. The $\tilde O$ notation
 hides polylogarithmic factors in these quantities and $1/\epsilon$, where $\epsilon$ is the desired
 precision of the output state. This is a quadratic improvement in $\beta$ and an exponential 
 improvement in $1/\epsilon$ with respect to a related algorithm presented in~\cite{PW09,CW10}.
 In fact, the main difference between our quantum algorithm and that of~\cite{PW09,CW10}
 is in the implementation of the operator $e^{-\beta H/2}$, where $H$ is the Hamiltonian of the system.
 Rather than using phase estimation, we use a technique introduced in~\cite{SOGKL02,BCC+14,CKS15} to decompose 
 $e^{-\beta H/2}$ as a linear combination of unitary operations and then apply results of spectral gap amplification in~\cite{SB13}
 to implement each such unitary. The same idea can be used to improve the running time of the algorithm presented in~\cite{BB10}.
 The second algorithm provides an estimate of $t_h$, the hitting time of a reversible, irreducible, and aperiodic Markov
 chain. It runs in time $\tilde O(1/(\epsilon \Delta^{3/2}))$, where $\epsilon$ is the absolute precision in the estimation
 and $\Delta$ is a parameter that satisfies $1/\Delta \ge t_h$. The $\tilde O$ notation hides factors 
 that are polynomial in $\log(1/(\epsilon \Delta))$ and $\log(N)$, where   
 $N$ is the dimension of the configuration
 space. In addition to the techniques used by the first algorithm,
 the second algorithm also uses recent methods for the quantum linear systems algorithm in~\cite{CKS15} and methods
 to estimate quantities at the so-called quantum metrology limit described in~\cite{KOS07}.

The paper is organized as follows. In Sec.~\ref{sec:maintools} we describe the main techniques
introduced in~\cite{SB13,BCC+14,BCC+15,CKS15,KOS07} that are also used by our algorithms. Then,
the quantum algorithm for the preparation of thermal Gibbs states of quantum systems
is described in Sec.~\ref{sec:Gibbs} and the quantum algorithm for estimating hitting times
of classical Markov chains is described in Sec.~\ref{sec:HT}.
We provide concluding remarks in Sec.~\ref{sec:conc}.

%%%%%%%%%%%%%%%%%%%%%%%%%%%%%%%%%%%%%%%%%%%
%%%%%%%%%%%%%%%%%%%%%%%%%%%%%%%%%%%%%%%%%%%%
\section{Main techniques}
\label{sec:maintools}
Our algorithms are based on techniques 
developed in the context of spectral gap amplification~\cite{SB13}, Hamiltonian simulation~\cite{BCC+14,BCC+15},
 quantum metrology~\cite{KOS07}, and solving linear systems of equations~\cite{CKS15}.
We first consider an arbitrary finite-dimensional quantum system modeled by a Hamiltonian $H$ that satisfies
\begin{align}
H \ket{\psi_j} = E_j \ket{\psi_j} \;.
\end{align}
 $E_j$ are the eigenenergies and $\ket{\psi_j}$ are the eigenstates, $j=0,1,\ldots N-1$, and $N$ is the dimension of the Hilbert space.
 We assume that $H$ describes a system of $n$ qubits and $N=2^n$~\cite{OGKL01,SOGKL02}.
 Furthermore, we assume that $H$ can be decomposed as
\begin{align}
H = \sum_{k=1}^{K} h_k \;,
\end{align}
where each $h_k \ge 0$ is a semidefinite positive Hermitian operator.
In some cases, the assumption on $h_k$ can be satisfied
after a simple rescaling of $H$ depending on its specification.

The results in~\cite{SB13} use the Hamiltonian
\begin{align}
\label{eq:GAP}
\tilde H =  \sum_{k=1}^{K} \sqrt{h_k} \otimes \left( \ket k \! \bra 0_{\rm a_1} + \ket 0 \! \bra k_{\rm a_1}\right) \; ,
\end{align}
where $\rm a_1$ refers to an ancillary qubit register of dimension $O( \log (K))$.
The important property is
\begin{align}
\label{eq:mainpropH}
 ( \tilde H  )^2 \ket{\phi} \otimes \ket 0_{\rm a_1} = \left( H \ket{\phi} \right) \otimes \ket 0_{\rm a_1} \; ,
\end{align}
for any $\ket \phi$.
Roughly, $\tilde H$ can be thought of as the square root of $H$.
Our algorithms will require evolving with $\tilde H$ for arbitrary time:
\begin{definition}
Let $\tilde W(t) := \exp(-i \tilde H t)$ be the evolution operator of $\tilde H$ for time $t$, and $\epsilon >0$ a precision parameter. We define
$W$ as a quantum circuit that satisfies $\| \tilde W(t) - W \| \le \epsilon$. The number of two-qubit gates
to implement $W$ (i.e., the gate complexity) is $C_W(t,\epsilon)$.
\end{definition}

When $H$ is a physical Hamiltonian described by local operators, $\tilde H$ may be efficiently
obtained with some classical preprocessing.
To obtain $C_W(t,\epsilon)$ in some instances,
we note that the results in~\cite{BCC+14,BCC+15} provide an efficient method for simulating Hamiltonians of complexity
polylogarithmic in $1/\epsilon$. In more detail,
we could assume that we have a presentation of the Hamiltonian as
\begin{align}
\label{eq:proj}
H= \sum_{k=1}^K \alpha_k \Pi_k\;,
\end{align}
 or
\begin{align}
\label{eq:Hunit}
H= \frac 1 2 \sum_{k=1}^K \alpha_k U_k\;,
\end{align}
where the coefficients satisfy $\alpha_k >0$. The operators $\Pi_k$ are projectors (i.e., $(\Pi_k)^2=\Pi_k$) and
$U_k$ are unitaries of eigenvalues $\pm 1$ in this case.
%where each $U_k$ is unitary and has eigenvalues $\pm 1$ and $\alpha_k>0$.
Many qubit Hamiltonians can be represented in this way, where the $U_k$ correspond, for example, to 
Pauli operators. We note that Eq.~\eqref{eq:Hunit} can be reduced to Eq.~\eqref{eq:proj}
by a simple rescaling in which $\Pi_k=(U_k+\one_N)/2$ and disregarding the factor proportional to the $N \times N$
identity operator $\one_N$. 
In either case, we assume that there is a mechanism
available to simulate $\Pi_k$ or $U_k$; that is, we assume access to a unitary
\begin{align}
\nonumber
Q &= - \sum_{k=1}^K e^{i \pi \Pi_k} \otimes \ket k \! \bra k_{\rm a_2} \\
& = \sum_{k=1}^K U_k \otimes \ket k \! \bra k_{\rm a_2} \;,
\end{align}
where $\rm a_2$ is also an ancillary register of $O(\log (K))$ qubits.
The gate complexity of each $U_k$  is $C_U$, which depends on the problem,
and the gate complexity of the conditional $U_k$ operation is $O(C_U \log (K))$.

Once the Hamiltonian $H$ has been reduced to the form of Eq.~\eqref{eq:proj},
we obtain
\begin{align}
\label{eq:Hgapamp}
\tilde H = \sum_{k=1}^K \sqrt{\alpha_k} \Pi_k \otimes \left( \ket k \! \bra 0_{\rm a_1} + \ket 0 \! \bra k_{\rm a_1}\right) \;.
\end{align}
To be able to use the results in~\cite{BCC+15} for simulating $\tilde H$ in this case, we note that
\begin{align}
\nonumber
 \ket k \! \bra 0_{\rm a_1} + \ket 0 \! \bra k_{\rm a_1} =& \frac i 2 \left[ e^{-i (\pi/2) (\ket k \bra 0_{\rm a_1} + \ket 0 \bra k_{\rm a_1} )} 
- \right. \\
\label{eq:unitdecomp}
& \left. - e^{i (\pi/2) (\ket k \bra 0_{\rm a_1} + \ket 0 \bra k_{\rm a_1} )} \right] \;.
\end{align}
This provides a decomposition of $\tilde H$ as a linear combination of $\tilde K=O(K)$ unitary operations $\tilde U_k$; that is,
\begin{align}
\tilde H = \sum_{k=1}^{\tilde K} \tilde \alpha_k \tilde U_k \;,
\end{align}
and $\tilde \alpha_k>0$. The unitaries in the right hand side of Eq.~\eqref{eq:unitdecomp}
can be implemented with $O(\log (K))$ two-qubit gates using standard techniques.
The algorithm in~\cite{BCC+15} assumes the ability to implement the unitary
\begin{align}
\tilde Q= \sum_{k=1}^{\tilde K} \tilde U_k \otimes \ket k \! \bra k_{\rm a_2} \;.
\end{align}
Since the unitaries $\tilde U_k$ are directly related
to the $U_k$, $\tilde Q$ can be simulated with $O(1)$ uses of $Q$
and additional two-qubit gates that do not contribute
significantly to the final gate complexity.

The query complexity of the method in~\cite{BCC+15} is determined by the number of uses of $\tilde Q$
to implement an approximation of $\tilde W(t)$. The gate complexity stated in~\cite{BCC+15} is the number of additional
two-qubit gates required.
Then, the results in~\cite{BCC+15} provide a Hamiltonian simulation method $W$
to approximate $\tilde W(t)$ for this case, within precision $\epsilon$, of query complexity
\begin{align}
O \left( \tau \log(\tau/\epsilon) / \log\log(\tau/\epsilon)\right) \;.
\end{align}
Here, $\tau = |t| \sum_{k} {\tilde \alpha_k}$ and thus $\tau =O( |t| \sum_{k} {\sqrt{ \alpha_k}})$.
%$\tau \le |t| \sqrt{K \sum_k \alpha_k}$.
The additional gate complexity of $W$ obtained in~\cite{BCC+15} for this case is
\begin{align}
O \left(K \tau \log(\tau/\epsilon) / \log\log(\tau/\epsilon)\right) \;.
\end{align}
These results also imply that the overall gate complexity
of $W$ is
\begin{align}
\label{eq:Wgatecomp}
C_W(t,\epsilon)=O \left(( \log (K) C_U +K) \tau \frac{\log(\tau/\epsilon)}{  \log\log(\tau/\epsilon)} \right) \;.
\end{align}
We refer to~\cite{BCC+15} for more details.

In general, our quantum algorithm to sample from Gibbs distributions provides
 an exponential improvement in terms of $1/\epsilon$,
with respect to other known algorithms~\cite{PW09,CW10,TOVPV11}, 
whenever $C_W(t,\epsilon)$
is polylogarithmic in $1/\epsilon$. As discussed,
this is the case for a large class of Hamiltonians such as those
when the $U_k$ are presented as Pauli operators, so that $C_U=O(n)$.

For the quantum algorithm that computes an estimate
of the hitting time of a Markov chain, we will assume that
we have query access to the Hamiltonian $H$, and that $H$ can be presented
as in Eq.~\eqref{eq:proj}. This assumes
the existence of a procedure that outputs the matrix elements of $H$.
Constructing a quantum circuit $W$ that approximates the evolution with $\tilde H$ in this 
case is  technically involved and we leave that analysis for Appx.~\ref{appx}.
As in the previous case, we use the methods in~\cite{BCC+14,BCC+15}
to show that $C_W(t,\epsilon)$ is almost linear in $|t|$
and sublogarithmic in $|t|/\epsilon$.

Another useful technique for our quantum algorithms,
also used in~\cite{SOGKL02,CW12,BCC+15,CKS15},
regards the implementation of linear combinations
of unitary operations. More specifically, assume that $X=\sum_{l=0}^{L-1} \gamma_l V_l$,
where $\gamma_l>0$ and $V_l$ are unitary operations, and that there is a mechanism
to implement $V_l$. That is, we have access to the unitary 
\begin{align}
R = \sum_{l=0}^{L-1} V_l \otimes \ket l \! \bra l_{\rm a_3} \; ,
\end{align}
where $\rm a_3$ is an ancillary register of $O(\log (L))$ qubits. Lemma 6 of~\cite{CKS15} implies
that we can prepare a normalized version of the state $X \ket \phi$ 
with $O(\gamma/\| X \ket \phi \|)$ uses of $R$ in addition to $O(L \gamma/\| X \ket \phi \|)$ two-qubit gates,
where $\gamma=\sum_{l=0}^{L-1}\gamma_l$. When $V_l=V^l$, for some unitary $V$, and the gate complexity of $V$ is $C_V$,
the gate complexity of $R$ is $O(L C_V)$. 
In this case, the overall gate complexity of the algorithm is $O(L C_V \gamma/\| X \ket \phi \|)$. 
This result follows from Lemma 8 of~\cite{CKS15}. The implication
is that the overall gate complexity is dominated by the largest gate complexity 
of the unitaries in $R$ times the number of amplitude amplification steps.

For completeness, the quantum algorithm to implement $X$ 
is built upon $O(\gamma/\| X \ket \phi \|)$ amplitude amplification steps~\cite{Gro96}.
The operation for state preparation starts by preparing the ancillary state 
\begin{align}
\label{eq:Bgate}
B \ket 0_{\rm a_3}= \frac 1 {\sqrt \gamma }\sum_{l=0}^{L-1} \sqrt{\gamma_l} \ket l_{\rm a_3} \;,
\end{align}
where $B$ is unitary. Applying $B$ requires $O(L)$ two-qubit gates and, in those cases
where we can exploit the structure of the coefficients $\gamma_l$,
it can be done more efficiently. The state preparation step then applies $R$ followed by $B^\dagger$.
One can show that the final state of this step is
\begin{align}
\left( \frac X \gamma  \ket \phi \right) \otimes \ket 0_{\rm a_3} + \sket {\Theta^\perp} \;,
\end{align}
where $\sket{\Theta^\perp}$ is supported in the subspace orthogonal to $\ket 0_{\rm a_3}$.
Amplitude amplification allows us to amplify the probability of observing the state $\ket 0_{\rm a_3}$
to a constant.
This state corresponds to the desired outcome. The number of amplitude amplification steps
is linear in the inverse of $\|(X/\gamma) \ket \phi \|$.

The third useful technique regards amplitude estimation~\cite{KOS07}.
Let $T$ be a unitary that implements
\begin{align}
T \ket \phi \ket 0 = (A \ket \phi) \ket 0 + \ket{\Phi^\perp} \ket 1 \;,
\end{align}
where $A$ is an operator that satisfies $\|A\| \le 1$ and $\| \ket{\Phi^\perp}\| \le 1$.
Our goal is to obtain an estimate of $\bra \phi A \ket \phi = \bra \phi \bra 0 T \ket \phi \ket 0$.
The results in~\cite{KOS07} imply that there exists a quantum algorithm that outputs
an estimate of the expectation value of $T$ within precision $\epsilon$.
For constant confidence level ($c \approx 0.81$),
the quantum algorithm uses $T$ and other two-qubit gates $O(1/\epsilon)$ times.
It also uses the unitary that prepares the initial state $\ket \phi$, $O(1/\epsilon)$ times.
Increasing the confidence level can be done with an additional overhead
that is logarithmic in $|1-c|$.

%%%%%%%%%%%%%%%%%%%%%%%%%%%%%%%%%%%%%%%%%%%%%%
%%%%%%%%%%%%%%%%%%%%%%%%%%%%%%%%%%%%%%%%%%%%%%
\section{Preparation of Gibbs states}
\label{sec:Gibbs}

The thermal Gibbs state of  a quantum system $H$ at inverse temperature $\beta \ge 0$ is
the density matrix
\begin{align}
\rho = \frac 1 {\cal Z} e^{-\beta H} \;,
\end{align}
where ${\cal Z}=\Tr[e^{-\beta H}]=\sum_j e^{-\beta E_j}$ is the partition function. Then, the probability
of encountering the system in the quantum state $\ket{\psi_j}$, after measurement, is $p_j=e^{-\beta E_j}/{\cal Z}$.
%We will assume that we have a presentation of $H$ such that there exists a method 
%to simulate

Given a precision parameter $\epsilon>0$, a quantum algorithm to sample
from the Gibbs distribution $p_j$ can be obtained from a unitary $\bar V$ that satisfies
\begin{align}
\label{eq:Gibbsapprox}
\Tr_{\rm a} \left[ \bar V  \left( \ket 0 \! \bra{0}  \otimes \ket 0 \! \bra{0}_{\rm a}   \right) \bar V^\dagger \right] = \hat \rho
\end{align}
and
\begin{align}
\label{eq:tracecond}
\frac 1 2 \| \hat \rho - \rho \|_1 \le \epsilon \; .
\end{align}
We use the label  ${\bf a}$  for an ancillary qubit system that will be discarded 
at the end of the computation. The dimension of ${\bf a}$ depends on the algorithm.
The requirement on the trace distance in Eq.~\eqref{eq:tracecond}
implies that no measurement can distinguish between $\rho$ and $\hat \rho$ 
with probability greater than $\epsilon$~\cite{Hel76}.

The main result of this section is:
\begin{theorem}
\label{thm:main1}
There exists a quantum algorithm that prepares an approximation of the Gibbs state. The quantum algorithm
implements a unitary $\bar V$ of gate complexity
\begin{align}
O\left (\sqrt{\frac N {\cal Z}} \left( C_W(t,\epsilon')+ n + \log (J) \right) \right) \; ,
\end{align}
with $t=O(\sqrt{\beta \log(1/\epsilon')})$, $\epsilon'=O(\epsilon \sqrt{{\cal Z}/N})$, and $J=O(\sqrt{\| H \|\beta}\log(1/\epsilon'))$.
%When $C_W(t,\epsilon')$ is polylogarithmic in $1/\epsilon'$, the gate complexity of $V$
%is polylogarithmic in $1/\epsilon$. 
\end{theorem}

When $H$ is presented as in Eq.~\eqref{eq:proj}, we can replace $C_W(t,\epsilon')$ by Eq.~\eqref{eq:Wgatecomp}
if we use the best-known Hamiltonian simulation algorithm.
%The overall gate complexity in this case is
%\begin{align}
%O\left (\sqrt{\frac N {\cal Z}} \left( (\log(K) C_U +K) \tau \frac{\log(\tau/\epsilon')}{  \log\log(\tau/\epsilon')}+ n + \log (J) \right) \right)
%%O\left( \sqrt{\frac {N } {\cal Z}} K (n+\log (K)C_U) \tau \frac{\log(\tau/\epsilon')}{  \log\log(\tau/\epsilon')} \right )
%\end{align}
%where $\tau \le t \sqrt{K \sum_k \alpha_k}$. 
In cases of interest, such as qubit systems given by Hamiltonians
that are linear combinations of Pauli operators, we have $\alpha_k=O(1)$, $C_U=O(n)=O(\log(N))$, and $K= O(n)=O(\log (N))$. In this case,
the overall gate complexity is
\begin{align}
O\left( \sqrt{\frac {N \beta } {\cal Z}} {\rm polylog}\left( \sqrt{\frac {N \beta} {\cal Z} } \frac 1 \epsilon \right) \right) \;.
\end{align}
%The complexity in $n$ may be reduced if we can exploit the structure of the Hamiltonian further.
The important result is that the complexity of our algorithm is polylogarithmic in $1/\epsilon$ and also improves upon
the complexity in $\beta$ with respect to the methods in~\cite{PW09,CW10}.

\vspace{0.5cm}
Our quantum algorithm to sample from Gibbs distributions uses the two techniques
discussed in Sec.~\ref{sec:maintools}.
In this case, we will be interested in implementing an operator proportional to $e^{-\beta H/2}$.
To find a decomposition as a linear combination of unitaries, we invoke the so-called Hubbard-Stratonovich transformation~\cite{Hub59}:
\begin{align}
\label{eq:HS}
e^{-\beta H/2} = \sqrt{\frac 1 {2 \pi}} \int_{-\infty}^{\infty} dy \; e^{-y^2/2} e^{-i y \sqrt{\beta H}} \;.
\end{align}
In our case, we do not have a method to simulate the evolution with $\sqrt{H}$.
Nevertheless, we assume that we can evolve with $\tilde H$, which satisfies Eq.~\eqref{eq:mainpropH}.
Then, 
\begin{align}
\label{eq:HS2}
(e^{-\beta H/2} & \ket \phi ) \otimes \ket 0_{\rm a_1}=
\\
\nonumber
& = \left( \sqrt{\frac 1 {2 \pi}} \int_{-\infty}^{\infty} dy \; e^{-y^2/2} e^{-i y \sqrt{\beta} \tilde H} \right) \ket \phi  \ket 0_{\rm a_1} \;,
\end{align}
for any state $\ket \phi$. Note that the ancilla $\rm a_1$ remains in the state $\ket 0_{\rm a_1}$ and will be discarded at the end
of the computation.
Equation~\eqref{eq:HS2} implies that the operator $e^{-\beta H/2}$ can be approximated by a linear combination
of evolutions under $\tilde H$. Because $y \in (-\infty , \infty)$, we will need to find an approximation
 by a finite, discrete sum of operators $e^{-i y_j \sqrt{\beta} \tilde H}$. We obtain:
%%%%%%%%%%
\begin{lemma}
\label{lem:firstdecomp}
Let 
\begin{align}
X'=\sqrt{\frac 1 {2 \pi}} \sum_{j=-J}^J \delta y \; e^{-y_j^2/2} e^{-i y_j \sqrt{\beta} \tilde H} \;,
\end{align}
where $y_j = j \delta y$, for some $J = \Theta( \sqrt{ \|  H \| \beta  }  \log(1/\epsilon') )$ and $\delta y = \Theta(1/\sqrt{ \|H \| \beta \log(1/\epsilon')})$. Then,
if $\|H \| \beta \ge 4$ and  $\log(1/\epsilon') \ge 4$,
\begin{align}
\label{eq:expapprox}
\| (e^{-\beta H/2} \ket \phi) \otimes \ket 0_{\rm a_1} - X' \ket \phi  \ket 0_{\rm a_1} \| \le \epsilon' /2\; ,
\end{align}
for all states $\ket \phi$.
\end{lemma}
\begin{proof}
Consider the real function
\begin{align}
f(\tilde x) = \frac 1 {\sqrt{2\pi}}  \sum_{j=-\infty}^\infty \delta y \; e^{-y_j^2/2} e^{-i y_j \sqrt{\beta} \tilde x}
\end{align}
where $\tilde x \in \mathbb R$ and assume $|\tilde x| \le a < \infty$. The Poisson summation formula and the Fourier transform of the Gaussian  imply
\begin{align}
f(\tilde x)=\sum_{k=-\infty}^\infty e^{-\omega_k^2/2} \;,
\end{align}
where $\omega_k = -\sqrt \beta \tilde x + k/\delta y$. Then, there exists 
\begin{align}
\delta y = \Omega \left( \frac 1 {a \sqrt{\beta}  + \sqrt{\log(1/\epsilon')}}\right)
\end{align}
such that $|f(\tilde x) - e^{-\omega_0^2/2}| \le \epsilon'/4$. Note that if $a \sqrt{\beta} \ge 2$ and $\sqrt{\log(1/\epsilon')} \ge 2$,
we can choose
\begin{align}
\delta y = \Theta \left( \frac 1 {a \sqrt{\beta \log(1/\epsilon')}}\right) \;.
\end{align}
Also,
\begin{align}
\nonumber
\left| \frac 1 {\sqrt{2\pi}}  \sum_{j=J}^\infty \delta y \; e^{-y_j^2/2} e^{-i y_j \sqrt{\beta} \tilde x} \right| & \le \frac{\delta y} {\sqrt{2\pi}} \sum_{j=J}^\infty   e^{-y_j^2/2} \\
\nonumber & \le \frac{\delta y} {\sqrt{2\pi}} \sum_{j=J}^\infty   e^{-y_J y_j/2} \\
\nonumber
& \le \frac{\delta y} {\sqrt{2\pi}} \frac{e^{-y_J^2/2}} {1-e^{-y_J \delta y/2}} \;.
\end{align}
It follows that there exists a value for $J$, which implies $y_J = \Theta(\sqrt{\log(1/\epsilon')})$, such that
\begin{align}
\left|  f(\tilde x) - \frac 1 {\sqrt{2\pi}}  \sum_{j=-J}^J \delta y \; e^{-y_j^2/2} e^{-i y_j \sqrt{\beta} \tilde x} \right| \le \epsilon'/4 \;.
\end{align}
Using the triangle inequality we obtain
\begin{align}
\left| e^{-\omega_0^2/2} - \frac 1 {\sqrt{2\pi}}  \sum_{j=-J}^J \delta y \; e^{-y_j^2/2} e^{-i y_j \sqrt{\beta} \tilde x} \right| \le \epsilon'/2 \;,
\end{align}
and we can represent
\begin{align}
e^{-\omega_0^2/2} = \frac 1 {\sqrt{2\pi}} \int_{-\infty}^{\infty} dy \; e^{-y^2/2} e^{-i y \sqrt{\beta} \tilde x} \;.
\end{align}

To prove the Lemma, it suffices to act with $X'$ on the eigenstates of $\tilde H$.
We can then use the previous bounds if we assume that $\tilde x$ denotes the corresponding eigenvalue of $\tilde H$.
In particular, $a= \| \tilde H\|$. Then, if $\| \tilde H \| \sqrt \beta \ge 2$ and  $\sqrt{\log(1/\epsilon')} \ge 2$,
it suffices to choose
\begin{align}
\delta y = \Theta \left( \frac 1 {\| \tilde H \| \sqrt{\beta  \log(1/\epsilon')}}\right) \label{eq:deltay}
\end{align}
and
\begin{align}
J = \frac{y_J}{\delta y}= \Theta \left( \log(1/\epsilon')  \| \tilde H \| \sqrt{\beta  }\right) \;. \label{eq:maxJ}
\end{align}
The result follows from noticing that  $\| \tilde H \| = O(\sqrt{\| H \|})$ and that $X'$ then approximates $e^{-\beta (\tilde H)^2/2}$.
Since we act on initial states of the form $\ket \phi \ket 0_{\rm a_1}$, the action of $e^{-\beta (\tilde H)^2/2}$ is the same as that of $e^{-\beta H/2}$
on these states.

%{\color{red}
%Anirban, prove this. Choose suitable values for $\delta y$ and $J$ as function of $\epsilon'$.}
\end{proof}

In general, we cannot implement the unitaries $e^{-i y_j \sqrt{\beta} \tilde H}$ exactly but we can do so
up to an approximation error. We  obtain:
\begin{corollary}
Let  $\log(1/ \epsilon') \ge 4$,  $\|H \| \beta \ge 4$, and $W_j$ be a unitary that satisfies
\begin{align}
\label{eq:Hamsimapprox}
\| W_j - e^{-i y_j \sqrt{\beta} \tilde H} \| \le \epsilon'/4 \;
\end{align}
for all $j=-J,-J+1,\ldots,J$.
Let 
\begin{align}
\label{eq:Xdef1}
X=\sqrt{\frac 1 {2 \pi}} \sum_{j=-J}^J \delta y \; e^{-y_j^2/2} W_j \;.
\end{align}
Then,
\begin{align}
\label{eq:expapprox2}
\| (e^{-\beta H/2} \ket \phi) \otimes \ket 0_{\rm a_1} - X \ket \phi  \ket 0_{\rm a_1} \| \le \epsilon' \;.
\end{align}
\end{corollary}
\begin{proof}
The coefficients in the decomposition of $X'$ in Lemma~\ref{lem:firstdecomp} satisfy
\begin{align}
\left | \sqrt{\frac 1 {2 \pi}} \sum_{j=-J}^J \delta y \; e^{-y_j^2/2} - 1 \right| \le \epsilon'/2 \le 1/4\;,
\end{align}
and thus
\begin{align}
\left | \sqrt{\frac 1 {2 \pi}} \sum_{j=-J}^J \delta y \; e^{-y_j^2/2}   \right| \le 5/4\;.
\end{align}
This follows from Eq.~\eqref{eq:expapprox} for the case of $H=\tilde H=0$. The triangle inequality and Eq.~\eqref{eq:Hamsimapprox} imply
\begin{align}
\| X - X'\| \le (5/16) \epsilon' \;,
\end{align}
and together with Eq.~\eqref{eq:expapprox} we obtain the desired result.
\end{proof}

In Sec.~\ref{sec:maintools} we described a technique to implement $X = \sum_{l=0}^{L-1} \gamma_l V_l$.
In this case, $l=j+J$ and $L=2J+1$. The coefficients and unitaries are $e^{-y_j^2/2}$ and $W_j$,
respectively. 
%Since $W_j$ are implemented using a Hamiltonian simulation method,
%they  satisfy $W_j=\bar W^j$ for some unitary $\bar W$.
The quantum algorithm for preparing Gibbs states will aim at preparing
a normalized version of $X \ket {\phi_0}$, for a suitable initial state $\ket {\phi_0}$, using the technique of Sec.~\ref{sec:maintools}.

%%%%%%%%%%%%%%%%%%%%%%%%%%%%%%%%%%%%%%%%%%%%
\subsection{Algorithm}
We set $\epsilon'=O(\epsilon \sqrt{{\cal Z}/N})$.
Our quantum circuit $\bar V$ is defined in two basic steps.
The first step regards the preparation of a maximally entangled state
\begin{align}
\ket{\phi_0} = \frac 1 {\sqrt N} \sum_{j=0}^{N-1} \ket{\psi_j} \otimes \ket{\psi_j^*}_{\rm a_4} \otimes \ket0_{\rm a_1, \rm a_2, \rm a_3}
\end{align}
where we used an additional ancillary system ${\rm a_4}$ of $n$ qubits. 
%All the other ancillary registers $\rm a_1$, $\rm a_2$, and $\rm a_3$
%are also initialized in $\ket 0$.  
$\rm a_1$, $\rm a_2$,  $\rm a_3$, and $\rm a_4$ build the ancillary register $\bf a$ of Eq.~\eqref{eq:Gibbsapprox}. Note that $\ket{\phi_0}$
coincides with the maximally entangled state
\begin{align}
\label{eq:simpleinitstate}
\ket{\phi_0} = \frac 1 {\sqrt N} \sum_{\sigma=0}^{N-1} \ket{\sigma} \ket{\sigma}_{\rm a_4}\otimes \ket0_{\rm a_1, \rm a_2, \rm a_3} \;,
\end{align}
where $\ket \sigma$ is a $n$-qubit state in the computational basis, i.e., $\ket \sigma=\ket{0 \ldots 0}, \ket{0 \ldots 1}, \ldots$.

The second step regards the preparation of a normalized version of $X \ket{\phi_0}$. This step
uses the algorithm for implementing linear combinations of unitary operations described in Sec.~\ref{sec:maintools}, which also
uses amplitude amplification.
The operator $X$ is defined in Eq.~\eqref{eq:Xdef1} and requires a Hamiltonian simulation method for implementing $W_j$.

%%%%%%%%%%%%%%%%%%%%%%%%%%%%%%%%%%%%%%%%%%%%%
\subsection{Validity and complexity}
As described, our quantum algorithm prepares the normalized state
\begin{align}
\frac{X \ket {\phi_0}} {\| X \ket {\phi_0} \|} \; ,
\end{align}
with constant probability. 
We also note that
\begin{align}
\label{eq:normdesired}
\| e^{-\beta H/2} \ket{\phi_0} \| = \sqrt{{\cal Z}/N}
\end{align}
and Eq.~\eqref{eq:expapprox2} implies
\begin{align}
\left|\| X \ket{\phi_0} \| - \sqrt{{\cal Z}/N} \right| = O(\epsilon \sqrt{{\cal Z}/N})
\end{align}
for our choice of $\epsilon'$. Then, the prepared state satisfies
\begin{align}
\label{eq:errorgibbs10}
\left \| \frac{X \ket {\phi_0}} {\| X \ket {\phi_0} \|} - \frac{e^{-\beta H/2} \ket {\phi_0}} {\| e^{-\beta H/2} \ket {\phi_0} \|} \right \| \le \epsilon/2 \;.
\end{align}

We note that
\begin{align}
\nonumber
&e^{-\beta H/2} \ket {\phi_0} = \\
&=\frac 1 {\sqrt N} \sum_{j=0}^{N-1} e^{-\beta E_j/2} \ket{\psi_j} \otimes \ket{\psi_j^*}_{\rm a_4} \otimes \ket0_{\rm a_1, \rm a_2, \rm a_3} \;.
\end{align}
If we disregard the ancillary system $\bf a$, this is the Gibbs state: The probability of obtaining $\ket{\psi_j}$, after measurement,
is proportional to $e^{-\beta E_j}$.
Then, the property of the trace norm being non-increasing 
under quantum operations and Eq.~\eqref{eq:errorgibbs10} imply
\begin{align}
\frac 1 2 \left \|  \hat \rho - \rho \right \| \le \epsilon \;,
\end{align}
where
\begin{align}
\hat \rho= \Tr_{\rm a} \left[ \frac {X \ket {\phi_0} \! \bra {\phi_0} X^\dagger} {\| X \ket {\phi_0} \|^2}  \right] \;;
\end{align}
see Eq.~\eqref{eq:Gibbsapprox}.
That is, $\hat \rho$ is the state prepared by our algorithm after tracing out the ancillary register $\bf a$.

The number of amplitude amplification steps
is $O(1/\|X \ket{\phi_0}\|)$ and Eq.~\eqref{eq:normdesired}
implies that this number is also $O(\sqrt{N/{\cal Z}})$. The gate complexity
of each step is the gate complexity of preparing $\ket{\phi_0}$ in addition
to the gate complexity of implementing $X$. The former is $O(n)$
as $\ket{\phi_0}$ takes the simple form of Eq.~\eqref{eq:simpleinitstate}
and can be prepared with $O(n)$ controlled operations. $X$ is implemented
in three stages as described in Sec.~\ref{sec:maintools}. The first stage requires the unitary $B$
used in Eq.~\eqref{eq:Bgate}. In this case, the coefficients $\gamma_l$ are proportional to
 $e^{-y_j^2/2}$. Then, the gate complexity of $B$ is $O(\log (J))$ in this case if we use 
 one of the methods developed in~\cite{KW08,Som15}. The second stage
 regards the implementation of $R$.
 In this example, $R$ is the unitary that implements $W_j$ conditional on the state $\ket j_{\rm a_3}$.
 Since $W_j$ corresponds to a Hamiltonian simulation algorithm that approximates evolutions with $\tilde H$,
 the gate complexity of $R$ is dominated by the largest gate complexity of $W_j$. In particular, $W_j$ 
 approximates $e^{-i \tilde H t}$ within precision $\epsilon'$ and for maximum $t=O(\sqrt{\beta \log(1/\epsilon')})$.
 Then the gate complexity of $R$ is order $C_W(t,\epsilon')$.

The overall gate complexity is then
\begin{align}
O\left (\sqrt{\frac N {\cal Z}} \left( C_W(t,\epsilon')+ n + \log (J) \right) \right) \;,
\end{align}
with $t=O(\sqrt{\beta \log(1/\epsilon')})$ and $\epsilon'=O(\epsilon \sqrt{{\cal Z}/N})$.
This proves Thm.~\ref{thm:main1}.

%%%%%%%%%%%%%%%%%%%%%%%%%%%%%%%%%%%%%%%%%%%%%
%%%%%%%%%%%%%%%%%%%%%%%%%%%%%%%%%%%%%%%%%%%%%
\section{Estimation of Hitting times}
\label{sec:HT}
We consider a stochastic process that models a Markov chain.
The number of different configurations is $N$ and $P$ is the $N \times N$
stochastic matrix. We label each configuration as $\sigma=0,1,\ldots,N-1$ and
the entries of $P$ are transition or conditional
probabilities $\Pr(\sigma'|\sigma)$. We will assume that $P$ is reversible and
irreducible, satisfies the so-called detailed balance condition~\cite{NB98}, and has 
nonnegative eigenvalues.
The unique fixed point of $P$ is the
$N$-dimensional
probability vector ${\bf \pi}$. It is useful to use the bra-ket notation, where $\ket \nu$
represents a vector ${\bf \nu} \in \mathbb C^N$ and $\bra v = (\ket v)^\dagger$. Then, $P \ket \pi = \ket \pi$,
\begin{align}
\ket \pi = \sum_\sigma \pi_\sigma \ket \sigma \doteq \begin{pmatrix} \pi_0 \cr \vdots \cr \pi_{N-1} \end{pmatrix} \;,
\end{align}
and $\pi_\sigma$ is the probability of finding configuration $\sigma$ when sampling from the fixed point of $P$.

The hitting time of a stochastic process is roughly defined as the first time
at which the process is encountered in a particular subset of configurations.
To define the hitting time in detail, we assume that there is a subset $\cM$
of $N_\cM$ configurations that are ``marked'' and the remaining 
$N_\cU$ configurations constitute the ``unmarked'' subset $\cU$. Here, $N_\cM+N_\cU=N$.
With no loss of generality, the stochastic matrix $P$ takes the form
\begin{align}
P = \begin{pmatrix}  P_{\cU \cU} & P_{\cU \cM } \cr
P_{\cM \cU} & P_{\cM \cM} 
\end{pmatrix}\;,
\end{align}
where $P_{\cU \cU}$ and $P_{\cM \cM}$ are matrices  (blocks) of dimension $N_\cU \times N_\cU$ and $N_\cM \times N_\cM$,
respectively, and $P_{\cM \cU}$ and $P_{\cU \cM}$ are rectangular blocks. 
The entries of the block $P_{\cS' \cS}$ determine the 
probability of a configuration being in the subset $\cS'$ given that the previous configuration
was in the subset $\cS$. Our assumptions imply $\cU ,\cM \ne \{ \emptyset \}$ and $P_{\cU \cM},P_{\cM \cU}\ne 0$.
%P_{ST} from S to T
The hitting time is the expected time to find a marked configuration if the initial probability vector is $\ket \pi$.
That is, as in~\cite{KOR10}, we define the hitting time of $P$ via the following classical algorithm:
%\begin{itemize}
\begin{enumerate}
\item Set $t=0$
\item Sample $\sigma$ from $\ket \pi$
\item If $\sigma \in \cM$, stop
\item Otherwise, assign $t \leftarrow t+1$, apply $P$, and go to 3.
%\item Go to 2.
\end{enumerate}
%\end{itemize}
The hitting time $t_h$ is  the expected value of the random variable $t$.
%to implement the fourth step of this algorithm.

We let $\ket{\pi_\cU}$ and $\ket{\pi_\cM}$ represent the probability vectors obtained
by conditioning $\ket \pi$ on $\cU$ and $\cM$, respectively. These are
\begin{align}
\ket{\pi_\cU} = \frac{\sum_{\sigma \in \cU } \pi(\sigma) \ket \sigma }{\pi_\cU} \; , \; \ket{\pi_\cM} = \frac{\sum_{\sigma \in \cM } \pi(\sigma) \ket \sigma }{\pi_\cM} \; ,
\end{align}
with $\pi_\cU=\sum_{\sigma \in \cU } \pi(\sigma)$ and $\pi_\cM=\sum_{\sigma \in \cM } \pi(\sigma)$.
It is useful to define the modified Markov chain
\begin{align}
P'= \begin{pmatrix}  P_{\cU \cU} & 0 \cr
P_{\cM \cU} & \one_{N_\cM}
\end{pmatrix} \;,
\end{align}
which refers to an ``absorbing wall'' for the subset $\cM$.  Here, $\one_{N_\cM}$ is the $N_\cM \times N_\cM$ identity matrix.
As defined, $P'$ does not allow for transitions from the subset $\cM$ to the subset $\cU$. 
We will observe below that $P'$, and thus $P_{\cU \cU}$, play an important role in the determination of $t_h$.

Our definition of hitting time implies
\begin{align}
\label{eq:ht1}
 t_h =\sum_{t=0}^{\infty } t \; \Pr(t) \;,
 \end{align}
 where $\Pr(t)$ is the probability of  $t$ if we use the previous classical algorithm.
In particular, $\Pr(t=0)=\pi_\cM$.
We rewrite
\begin{align}
\label{eq:ht2}
 t_h = \sum_{t'=0}^{\infty }   \Pr(t > t')
 \end{align}
  so that we take into account the factor $t$ in Eq.~\eqref{eq:ht1}, i.e., $\Pr(t>t')=\Pr(t'+1)+\Pr(t'+2)+\ldots$.
Note that 
\begin{align}
\nonumber
 \Pr(t > t') & =\pi_\cU \bra{1_\cU} (P')^{t'} \ket{\pi_\cU} \\
 \label{eq:prht}
& =\pi_\cU \bra{1_\cU} (P_{\cU \cU})^{t'} \ket{\pi_\cU} \;,
\end{align}
where $\ket{1_\cU}= \sum_{\sigma \in \cU} \ket \sigma$. This is because, conditional on $t>0$, which occurs with probability $\pi_\cU$,
the initial probability vector is proportional to $(P_{\cU \cU})^{t'}\ket{\pi_\cU}$.
 Then, the probability
of having $t >t'$ is measured by the probability of remaining in $\cU$ after $P_{\cU \cU}$ was applied $t'$ times.
%, assuming that we evolved with $P'$(to make sure we never were in $M$).
Equations~\eqref{eq:ht2} and \eqref{eq:prht} imply
\begin{align}
\label{eq:ht3}
t_h = \pi_\cU \bra{1_\cU} ( {\one_{N_\cU}-P_{\cU \cU}})^{-1} \ket{\pi_\cU} \; ,
\end{align}
where we used $(1-x)^{-1}=\sum_{t'=0}^{\infty} x^{t'}$. 
We note that $\one-P_{\cU \cU}$ is invertible under our assumptions, since 
the eigenvalues of $P_{\cU \cU}$ are strictly smaller than 1 (see below).

The complexity of a method that estimates $t_h$
using the previous classical algorithm
also depends on the variance of the random variable $t$.
This is
\begin{align}
\sigma^2 = \sum_{t=0}^\infty t^2 \Pr(t) -(t_h)^2 \;,
\end{align}
and after simple calculations, we can rewrite it as
\begin{align}
\sigma^2 =2 \pi_\cU \bra{1_\cU}  \sum_{t'=0}^{\infty } t' (P_{\cU \cU})^{t'} \ket{\pi_\cU}+ t_h - (t_h)^2 \;.
\end{align}
For constant confidence level and precision $\epsilon$ in the estimation of $t_h$, Chebyshev's inequality
implies that the previous classical algorithm must be executed $M=O((\sigma/\epsilon)^2)$ times
to obtain $t_1,\ldots,t_M$ and estimate $t_h$ as the average of the $t_i$. The expected
number of applications of $P$ is then $M t_h=O(t_h(\sigma/\epsilon)^2 )$.

To bound the classical complexity, we consider the worst case scenario
in which $\ket{\pi_\cU}$ is an eigenvector of $P_{\cU \cU}$ corresponding to its largest eigenvalue $1-\Delta <1$.
%(Note that, under our assumptions, $\|P_{\cU \cU}\|<1$.)
In that case, $t_h=\pi_\cU/\Delta$ and $\sigma^2=O(\pi_\cU/\Delta^2)$. When $\Delta \ll 1$,
the expected number of applications of $P$ is then $O(1/(\Delta^3 \epsilon^2))$ in this case. 
This determines the average complexity of the classical algorithm that estimates $t_h$.

The entries of the symmetric discriminant matrix $S$ of $P$ are
\begin{align}
S_{\sigma \sigma'}=\sqrt{\left(P \circ P^\dagger\right)_{\sigma \sigma'}} \;,
\end{align}
where $\circ$ is the Hadamard product.
The detailed balance condition implies
\begin{align}
\pi(\sigma') \Pr(\sigma|\sigma') = \Pr(\sigma'|\sigma) \pi(\sigma) 
\end{align}
and thus
\begin{align}
\sqrt{ \Pr(\sigma|\sigma') \Pr(\sigma'|\sigma)} = \Pr(\sigma|\sigma') \sqrt{\frac{\pi(\sigma')}{\pi(\sigma)}} \;.
\end{align}
Then,
\begin{align}
\label{eq:discriminant}
S=   D^{-1} P D 
  \;,
\end{align}
where $D$ is a diagonal matrix of dimension $N$
with entries given by $\sqrt{\pi(\sigma)}$. 
The symmetric matrix or Hamiltonian $\bar H = \one_N - S$ is known to be ``frustration free''~\cite{SBO07} and can be represented
as in Eq.~\eqref{eq:proj} using a number of techniques. For example, if $P$ has at most $d$ nonzero entries
per row or column (i.e., $P$ is $d$-sparse),
the number of terms $K$ in the representation of $\bar H$ can be made linear or quadratic in $d$; see Appx.~\ref{appx}
or~\cite{BOS15}
for more details.

% In this case,
%evolutions with $h_k$ or $\sqrt{h_k}$ may be implemented efficiently using the techniques developed in the 
%context of Hamiltonian simulation~\cite{BCC+14,BCC+15}. This result also holds for other representations of $P$.
%We refer to~\cite{BOS15} for more details.

We now let $\Pi_\cU$ be the projector into the subset $\cU$ and
define
\begin{align}
H = \Pi_\cU \bar H \Pi_\cU \;.
\end{align}
Note that 
\begin{align}
\label{eq:discriminantU}
H = \one_{N_\cU} - D^{-1}_\cU P_{\cU \cU} D_{\cU} \;,
\end{align}
where $D_{\cU}$ is the diagonal matrix obtained by projecting $D$ of Eq.~\eqref{eq:discriminant}
into the subspace $\cU$. That is, $D_\cU= \Pi_\cU D \Pi_\cU$ and Eq.~\eqref{eq:discriminantU} implies $H>0$.
%$H$ can be obtained from $\one_{N_\cU}$ minus a similarity
%transformation of $P_{\cU \cU}$ [Eq.~\eqref{eq:discriminant}], and thus $H>0$.
%$\| H \| \le \| \bar H \| = \|P\| =1$ and also $H >0$.
Then, Eqs.~\eqref{eq:discriminant} and~\eqref{eq:ht3} imply
\begin{align}
\label{eq:ht5}
t_h=\pi_\cU \bra{\sqrt{\pi_\cU}} ( 1/H) \sket{\sqrt{\pi_\cU}} \; ,
\end{align}
where we defined $\sket{\sqrt{\pi_\cU}}=\sum_{\sigma \in \cU} \sqrt{\pi(\sigma)} \ket \sigma/\sqrt{\pi_\cU}$ so that
 $\sket{\sqrt{\pi_\cU}}$ is normalized according to the Euclidean norm.
 A similar expression for $t_h$ was obtained in~\cite{KOR10}.

In Appx.~\ref{appx} we describe how $H$
can be specified as $H= \sum_{k=1}^K \alpha_k \Pi_k$,
where $\alpha_k>0$ and $\Pi_k$ are projectors.
%
%Assuming a presentation of $\bar H= \sum_{k=1}^K \bar \alpha_k \bar \Pi_k$,
%where $\bar \alpha_k>0$ and $\bar \Pi_k$ are projectors, we obtain
%\begin{align}
%\nonumber
%H &= \sum_{k=1}^K \alpha_k \Pi_k \\
%&  = \sum_{k=1}^K \bar \alpha_k \Pi_U \Pi_k \Pi_U \;.
%\end{align}
Then $H$ is of the form of Eq.~\eqref{eq:proj} and we write $\tilde H$
for the associated Hamiltonian according to Eq.~\eqref{eq:GAP}.
$C_W(t,\epsilon)$ is the complexity of 
approximating $\tilde W(t)=\exp(-i \tilde H t)$, and we roughly describe
a method for simulating $\tilde W(t)$ below.

A quantum algorithm to obtain $t_h$ 
can be constructed from the relation in Eq.~\eqref{eq:ht5}.
That is, $t_h$ coincides with $\pi_\cU$ times the expected value of the operator
$1/H$ in the pure state $\ket{\sqrt{\pi_\cU}}$. For our quantum algorithm, we also assume that there is a 
unitary procedure (oracle) $Q_\cU$ that allows us to implement the transformation 
\begin{align}
Q_\cU \ket \sigma = - \ket \sigma \ {\rm if} \ \sigma \in \cU
\end{align}
and $Q_\cU \ket \sigma=\ket \sigma$ otherwise. We also assume access
to a unitary $Q_\spi$ such that
\begin{align}
Q_\spi \ket 0 = \ket{\sqrt \pi} \;.
\end{align}
We write $C_\cU$ and $C_\spi$ for the respective gate complexities.
The main result of this section is:
%%%%%%%%%%%%%%%%%%%%%%%%%%%%%%%%%%%%%%%%%%%%%
\begin{theorem}
\label{thm:main2}
There exists a quantum algorithm to estimate $t_h$ within precision $\epsilon$
and constant confidence level
that implements a unitary $\bar V$ of gate complexity
\begin{align}
O \left( \frac 1 {\epsilon'} \left( C_W(t,\epsilon') + C_\cU + C_\spi +C_B\right) \right) \;,
 \end{align} 
where $C_B= O(\log(1/(\Delta \epsilon)))$, $\epsilon'=O(\epsilon\Delta/\log(1/(\epsilon \Delta))$,
and $t=O\left(\log(1/(\epsilon\Delta))/\sqrt{\Delta}\right)$.
\end{theorem}
%%%%%%%%%%%%%%%%%%%%%%%%%%%%%%%%%%%%%%%%%%
In Appx.~\ref{appx} we describe a method to simulate the evolution with $\tilde H$.
To this end,
we also assume that there exists a procedure $Q_P$ that computes the locations 
and magnitude of the nonzero entries
of the matrix $P$. More specifically, $Q_P$ performs the map
\begin{align}
Q_P  \ket \sigma  \rightarrow & \ket \sigma \ket{\sigma'_1,\ldots,\sigma'_d} \otimes
\\
\nonumber
&\otimes \ket{\Pr(\sigma|\sigma'_1),\Pr(\sigma'_1|\sigma),\ldots
,\Pr(\sigma|\sigma'_d),\Pr(\sigma'_d|\sigma) } \;,
\end{align}
where $d$ is the sparsity of $P$. The configurations
$\sigma'_i$ are such that $\Pr(\sigma|\sigma'_i),\Pr(\sigma'_i|\sigma)\ne 0$.
We write $C_P$ for the complexity of implementing $Q_P$.
In Appx.~\ref{appx} we describe a decomposition of $\tilde H=\sum_{k=1}^{\tilde K}   \tilde U_k/2$ in terms of $\tilde K=O(d^2)$ unitaries,
so that we can use the results of~\cite{BCC+15} to simulate $\tilde W (t)=\exp(-i \tilde H t)$.
Each $\tilde U_k$ can be implemented with $O(1)$ uses of $Q_\cU$ and $Q_P$,
and $O(d \log (N))$ additional gates. 
%The coefficients $\tilde \alpha_k=O(1)$ are simple to calculate and can be made uniform.
Using the results of~\cite{BCC+15} and Sec.~\ref{sec:maintools}, the complexity for simulating $\tilde W(t)$ within precision $\epsilon'$,
obtained in Eq.~\eqref{eq:finalWcomp}, is
\begin{align}
C_W(t,\epsilon')= O \left( (d \log (N) + C_P+C_\cU) \frac{\tau \log(\tau/\epsilon')}{  \log\log(\tau/\epsilon')}\right) \;, \label{eq:estHcomplexty1}
\end{align}
where $\tau = |t| d^2$. Note that $C_W(t,\epsilon')$ is almost linear in $|t|$ and polynomial in $d$, and the
dependence on $d$ may be improved by using the results in~\cite{BCK15}.
Then, assuming access to $Q_P$, we obtain:
\begin{corollary}
There exists a quantum algorithm to estimate $t_h$ within precision $\epsilon$
and constant confidence level
that implements a unitary $\bar V$ of gate complexity
\begin{align}
 \tilde O \left(   \frac 1 {\epsilon \Delta} \left( \frac{d^2}{\sqrt \Delta} (d \log (N) + C_P +C_\cU)  + C_\spi \right) \right) \;.
\end{align}
The $\tilde O$ notation hides factors that are polylogarithmic in $d/(\epsilon \Delta)$.
%{\color{red} We should state that as well}
\end{corollary}

The dominant scaling of the complexity in terms of $\Delta$ and $\epsilon$
is then $\tilde O(1/(\epsilon \Delta^{3/2}))$, which is a quadratic improvement over the classical complexity
obtained above.

%{\color{red} In contrast with the previous section
%here we cannot have good scaling with $\epsilon$.
%Our goal is to have good scaling with $\Delta$.}

\vspace{0.5cm}
%{\color{red} In contrast with the previous section
%here we cannot have good scaling with $\epsilon$.
%Our goal is to have good scaling with $\Delta$.}

Our quantum algorithm uses the three techniques
described in Sec.~\ref{sec:maintools} and uses some other
results in~\cite{CKS15}. In fact,
since $H >0$, we can improve some results in~\cite{CKS15}
that regard the decomposition of the inverse of a matrix as a linear combination of unitaries.
That is, for a positive matrix $H$, we can use the identity
\begin{align}
%\nonumber
\frac 1 H & =\frac 1 2  \int_0^\infty d\beta \; e^{-\beta H/2} \;,
%\\
%& = \frac 1 {\sqrt{2\pi}}\int_0^\infty dz \int_{-\infty}^{\infty} dy e^{-y^2/2} e^{-i y \sqrt{2 z} \sqrt H } \;,
\end{align}
and use the Hubbard-Stratonovich transformation of Eq.~\eqref{eq:HS2} to simulate $e^{-\beta H/2}$.
%where the second equality follows from the Hubbard-Stratonovich transformation.
Roughly, $1/H$ can be simulated  by a linear combination of unitaries, 
each corresponding to an evolution with $\tilde{H}$ for time $y \sqrt{ \beta}$.
Since $\beta \in[0,\infty)$ and $y \in (-\infty,\infty)$, we will need to find an approximation by a finite, discrete sum of operators $e^{-i y_j \sqrt{\beta_k} \tilde H }$.
We obtain:
%%%%%%%%%%%%%%%%%%%%%%%%%
\begin{lemma}
\label{lem:2algapprox}
Let 
\begin{equation}
X' = \frac 1 {\sqrt{2\pi}} \sum_{k=0}^K {\delta z } \sum_{j=-J}^J \delta y \; e^{-y_j^2/2} e^{-i y_j \sqrt{z_k} \tilde H } \;,
\end{equation}
where $y_j= j \delta y$, $z_k = k \delta z$, and $\| H \| \ge \Delta >0$. Then, there exists
 $J= \Theta ( \sqrt{1/\Delta}\log^{3/2}(1/(\Delta \epsilon)) )$,  $K = \Theta ((1/\Delta) \log(1/(\Delta \epsilon))/\epsilon)$,
 $\delta y= \Theta ( \sqrt{\Delta}/   \log(1/(\Delta \epsilon)))$, and $\delta z=\Theta(\epsilon)$
 such that
% 
% 
%for some $J= \Theta \left( \sqrt{\kappa}\log^{3/2}(\kappa/\epsilon') \right)$, $K = \Theta (\kappa^2 \log(\kappa/\epsilon')/\epsilon')$, $\delta y= \Theta \left( 1/ (\sqrt{\kappa}  \log(\kappa/\epsilon')\right)$ and $\delta z=\Theta (\epsilon'/\kappa)$.
%Assume $\|H\| \ge 1/\kappa$ for some  $1 \le \kappa <\infty$.
%Then, 
\begin{align}
\left \| \left (\frac 1 H \ket \phi \right) \otimes \ket 0_{\rm a_1} - X' \ket \phi  \ket 0_{\rm a_1} \right \| \le \epsilon/2 
\end{align}
for all states $\ket \phi$. \label{lem:2nddecomp}
\end{lemma}
\begin{proof}
We first consider the approximation of $1/x$ by a finite sum of $e^{-z_k H}$:
\begin{align}
\left| \frac 1 x - \delta z \sum_{k=0}^{K-1} e^{-z_k x} \right| &= \left| \frac 1 x - \delta z \frac {1- e^{-z_K x}}{1-e^{-\delta z x}} \right| \\
& = \left| \frac 1 x O(\delta z x + e^{-z_K x})\right| 
\end{align}
Assuming that $1\geq x\geq 1/\kappa$ so that $1/x \le \kappa$, we can upper bound the above 
quantity by $\epsilon/4$ if we choose $e^{-z_K /\kappa} = \Theta(\epsilon/\kappa)$ and $\delta z  = \Theta(\epsilon)$.
These imply
\begin{align}
 z_K = \Theta(\kappa \log(\kappa/\epsilon))
\end{align}
and
\begin{align}
K = z_K /\delta z= \Theta (\kappa \log(\kappa/\epsilon)/\epsilon) \;.
\end{align}

In the next step, we invoke the proof of Lemma~\ref{lem:firstdecomp} and approximate each $e^{-z_k x}$ as 
\begin{equation}
g_k(x) = \frac{1}{\sqrt{2\pi}}\sum_{j=-J}^J \delta y \; e^{-y_j^2/2} e^{-i y_j \sqrt{ z_k x}} \;,
\end{equation}
and we need to choose $J$ and $\delta y$ so that the approximation error is bounded by $\epsilon/(4z_K)$. 
Then,
\begin{align}
\nonumber
\left| \delta z \sum_{k=0}^{K-1} (e^{-z_k x} -g_k(x))  \right| &\le (\delta z K) \frac {\epsilon}{4 z_K}\\
& \le \epsilon/4 \;.
\end{align}
Lemma~\ref{lem:firstdecomp} then implies
\begin{align}
\nonumber
\delta y &= \Theta \left( \frac 1 { \sqrt{z_K  \log(z_K/\epsilon)}}\right) \\
\nonumber
& = \Theta \left( \frac 1 { \sqrt{\kappa \log(\kappa/\epsilon)  \log(\kappa \log(\kappa/\epsilon)/\epsilon)}}\right)
\\
& = \Theta \left( \frac 1 { \sqrt{\kappa} \log(\kappa/\epsilon)}  \right) \;,
\end{align}
and
\begin{align}
\nonumber
J &= \Theta\left( \sqrt{z_K} \log(z_K/\epsilon) \right) \\
& =  \Theta \left( \sqrt{\kappa} \log^{3/2}(\kappa/\epsilon) \right) \;.
\end{align}

Thus far, we presented an approximation of $1/x$, for $1 \ge x \ge 1/\kappa$, as a doubly weighted sum of 
terms $\exp(-iy_j\sqrt{z_k x})$. To obtain the desired result, it suffices to act with $X'$ on any eigenstate
of $\tilde H$ and replace $\sqrt x$ by the corresponding eigenvalue, as we did in Lemma~\ref{lem:firstdecomp}. Since $H \ge \Delta$,
we need to replace $\kappa$ by $1/\Delta$ in the bounds obtained for $\delta z$, $\delta y$, $K$, $J$, $z_K$, and $y_J$.
%We can act with $X'$ on any state expressed as a sum of eigenstates of $\tilde{H}$. The proof follows.
\end{proof}
In general, we cannot implement the unitaries $e^{-i y_j \sqrt{2 z_k} \tilde H }$ exactly but we can do so
up to an approximation error. We  obtain:
\begin{corollary}
Let $\epsilon >0$ and $W_{jk}$ be a unitary that satisfies
\begin{align}
\label{eq:Hamsimapprox2}
\| W_{jk} - e^{-i y_j \sqrt{2 z_k} \tilde H } \| \le \frac{\epsilon}{4 z_K} \;.
\end{align}
Let 
\begin{align}
\label{eq:Xdef2}
X= \frac 1 {\sqrt{2\pi}} \sum_{k=0}^K \delta z \sum_{j=-J}^J \delta y \; e^{-y_j^2/2} W_{jk} \;.
\end{align}
Then,
\begin{align}
\label{eq:expapprox5}
\left \| \left (\frac 1 H \ket \phi \right) \otimes \ket 0_{\rm a_1} - X \ket \phi  \ket 0_{\rm a_1} \right \| \le \epsilon \;.
\end{align}
\end{corollary}
\begin{proof}
If we replace each $e^{-i y_j \sqrt{z_k x}}$ by a 
term that is an $\epsilon/(2z_K)$ approximation in the definition of $g_k(x)$, it yields an approximation
of $1/x$ within precision $\epsilon/2$  plus
\begin{align}
\frac{\delta z K }{\sqrt{2\pi}} \sum_{j=-J}^J \delta y e^{-y_j^2/2}( \epsilon/(4z_K) ) =\frac{ \epsilon}{4\sqrt{2\pi}} \sum_{j=-J}^J \delta y e^{-y_j^2/2}\;.
\end{align}
Our choice of parameters in Lemma~\ref{lem:2algapprox} implies
\begin{align}
\label{eq:sumestimate}
\left| 1- \frac{ 1}{\sqrt{2\pi}} \sum_{j=-J}^J \delta y e^{-y_j^2/2}\right| \le \epsilon/4
\end{align}
so that the additional error is bounded by $\epsilon/2$.
The proof follows by replacing $\sqrt x$ by the corresponding eigenvalue of $\tilde H$
and $g_k(x)$ by the linear combination of the $W_{jk}$ with  weights $\delta y e^{-y_j^2/2}$.
\end{proof}

So far we showed that $1/H$ can be approximated within precision $\epsilon$
by a linear combination of unitaries that correspond to evolutions under $\tilde H$
for maximum time $y_J \sqrt{z_K}=\Theta((1/\sqrt \Delta) \log(1/(\Delta \epsilon)))$.
Each such evolution must be implemented by a method for Hamiltonian simulation
that approximates it within precision $\epsilon'=O(\epsilon/z_K)=O(\epsilon \Delta/(\log(1/(\epsilon \Delta))))$.
%At the end, it will suffice to choose $\epsilon'=\epsilon$.

In Sec.~\ref{sec:maintools} we described a technique to implement $X = \sum_{l=0}^{L-1} \gamma_l V_l$.
In this case, the coefficients and unitaries are $\delta z \delta y e^{-y_j^2/2}/\sqrt{2 \pi}$ and $W_{jk}$,
respectively. From Lemma~\ref{lem:2algapprox} and Eq.~\eqref{eq:sumestimate}, it is simple to show
\begin{align}
\nonumber
\left|z_K - \gamma \right|&\le \sum_{k=0}^K \delta z \left| 1- \frac{1}{\sqrt{2 \pi}} \sum_{j=-J}^{J} \delta y e^{-y_j^2/2}\right| \\
\label{eq:gammaestimate}
& \le z_K \epsilon/4
\end{align}
and thus $\gamma \approx z_K$ or $\gamma = \Theta((1/\Delta) \log(1/(\Delta \epsilon)))$.

Last, we define the unitary $T=(T_2)^\dagger T_1$ such that
\begin{align}
T_1 \ket 0 \ket 0_{\rm a} = \sqrt{\pi_\cU} \frac {X'} \gamma \ket{\sqrt{\pi_\cU}} \ket 0_{\rm a}  + \sket{\Theta^\perp} \;,
\end{align}
and  $\sket{\Theta^\perp}$ is supported in the subspace orthogonal to $\ket 0_{\rm a}$.
The ancillary register $\bf a$ includes the ancillary registers $\rm a_1, \rm a_2, \rm a_3$ as needed
for evolving with $\tilde H$ and implementing $X'$. That is, $\ket 0_{\rm a}=\ket 0_{\rm a_1, \rm a_2, \rm a_3}$.
$T_1$ can be implemented as follows. It first uses $Q_\spi$ to prepare the quantum state $\ket{\sqrt{\pi}}$. It then uses $Q_\cU$
to prepare $\sqrt{\pi_\cU} \ket{\sqrt{\pi_\cU}} \ket 0_{\rm a'} + \sqrt{\pi_\cM} \ket{\sqrt{\pi_\cM}} \ket 1_{\rm a'}$, where the ancilla
qubit $\rm a'$ is part of the register $\bf a$. Then, conditional on $\ket 0_{\rm a'}$,
it implements $X'/\gamma$ as discussed in Sec.~\ref{sec:maintools}. $T_2$ is the unitary that prepares
$\sqrt{\pi_\cU} \ket{\sqrt{\pi_\cU}} \ket 0_{\rm a}  + \sqrt{\pi_\cM} \ket{\sqrt{\pi_\cM}} \sket {0^\perp}_{\rm a}$,
where $\sket {0^\perp}_{\rm a}$ is orthogonal to $\ket {0}_{\rm a}$. Then, if $T=(T_2)^\dagger T_1$, we obtain
\begin{align}
\label{eq:ht9}
\bra 0 \bra 0_{\rm a} T \ket 0 \ket 0_{\rm a} = \frac {\pi_\cU}{\gamma} \bra{\sqrt{\pi_\cU}} X'  \ket{\sqrt{\pi_\cU}} \;.
\end{align}

%%%%%%%%%%%%%%%%%%%%%%%%%%%%%%%%%%%%%%%%%%%
%%%%%%%%%%%%%%%%%%%%%%%%%%%%%%%%%%%%%%%%%%%
\subsection{Algorithm}
We set $\epsilon'=O(\epsilon \Delta/\log(1/(\epsilon \Delta)))$.
The quantum algorithm for estimating the hitting time consists of two basic steps.
The first step uses the amplitude estimation algorithm of~\cite{KOS07} to provide an estimate of
$\bra 0 \bra 0_{\rm a_3} T \ket 0 \ket 0_{\rm a_3}$ within precision $\epsilon '$ and constant confidence level ($c \approx 0.81$). Call that estimate $\tilde t_h$.
The output of the algorithm is $\hat t_h =z_K \tilde t_h$.

%%%%%%%%%%%%%%%%%%%%%%%%%%%%%%%%%%%%%%
%%%%%%%%%%%%%%%%%%%%%%%%%%%%%%%%%%%%%%
\subsection{Validity and complexity}
%We disregard the complexity of state preparation and oracle use?
As described, our quantum algorithm provides a $O(\epsilon' z_K)$
estimate of $z_K \bra 0 \bra 0_{\rm a_3} T \ket 0 \ket 0_{\rm a_3}$.
Using Eqs.~\eqref{eq:ht5} and~\eqref{eq:ht9}, the output is an estimate
of $(z_K/\gamma) t_h$ within precision $O((z_K/\gamma) \epsilon + z_K \epsilon')$.
Our choice of $\epsilon'$ implies that this is $O((z_K/\gamma) \epsilon)$. Also, using Eq.~\eqref{eq:gammaestimate},
we obtain
\begin{align}
\left| 1 - \frac {z_K} {\gamma} \right| = O(\epsilon).
\end{align}
Then, our quantum algorithm outputs $\hat t_h$, an estimate of $t_h$
within absolute precision $O(\epsilon)$.

Our quantum algorithm uses $T$, $O(1/\epsilon')$ times. Each $T$ uses $Q_\cU$ and $Q_\spi$ two times,
in addition to the unitaries needed to implement $X'$. Each such unitary requires evolving with $\tilde H$
for maximum time $t=O((1/\sqrt \Delta) \log(1/(\Delta \epsilon)))$. In addition, each such unitary
requires preparing a quantum state proportional to
\begin{align}
\frac 1 {\sqrt {\gamma}} \sum_{j,k}\left (\frac{\delta y \delta z e^{-y_j^2/2}}{\sqrt{2 \pi}}\right)^{1/2} \ket{j,k} \;.
\end{align}
The gate complexity for preparing this state using the results in~\cite{KW08,Som15} 
is $C_B=O( \log (J)+\log (K) )$ and then $C_B= O(\log(1/(\Delta \epsilon)))$.
The overall gate complexity is
\begin{align}
O \left( \frac 1 {\epsilon'} \left( C_W(t,\epsilon') + C_\cU + C_\spi +C_B\right) \right) \;.
\end{align}
This proves Thm.~\ref{thm:main2}.
Using Eq.~\eqref{eq:finalWcomp} and replacing for $\epsilon'$ and $t$, 
and disregarding terms that are polylogarithmic in $d/(\epsilon \Delta)$, the gate complexity is
\begin{align}
 \tilde O \left(   \frac 1 {\epsilon \Delta} \left( \frac{d^2}{\sqrt \Delta} (d \log (N) + C_P +C_\cU)  + C_\spi  \right) \right) \;.
 \end{align}

%%%%%%%%%%%%%%%%%%%%%%%%%%%%%%%%%%%%%%%%%%%%%%%
%%%%%%%%%%%%%%%%%%%%%%%%%%%%%%%%%%%%%%%%%%%%%%%
\section{Conclusions}
\label{sec:conc}

We provided quantum algorithms for solving two problems of stochastic processes, namely the preparation of a thermal Gibbs state of a quantum system
and the estimation of the hitting time of a Markov chain. Our algorithms combine many techniques, 
including Hamiltonian simulation, spectral gap amplification, and methods for the quantum linear systems algorithm.
They provide significant speedups with respect to known classical and quantum algorithms for these problems
and are expected to be relevant to research areas in statistical physics and computer science, including optimization
and the design of search algorithms.

We first showed that, starting from a completely entangled state, we can prepare a state that is $\epsilon$-close (in trace distance) to a thermal Gibbs state using resources that scale polylogarithmic in $1/\epsilon$. This is an exponential improvement over previously known algorithms that rely on phase estimation and have complexity that depends polynomially in $1/\epsilon$~\cite{PW09,CW10}. Our algorithm circumvents the limitations of phase estimation by approximating the exponential operator as a finite linear combination of unitary operations and using techniques developed in~\cite{BCC+15} to implement it. We also used techniques developed in the context of spectral gap amplification~\cite{SB13} to improve the complexity
dependence on the
inverse temperature, from almost linear in $\beta$ to almost linear in $\sqrt \beta$.

Next, we presented a quantum algorithm to estimate the hitting time of a Markov chain, initialized in its stationary distribution, with almost quadratically less resources in all parameters than a classical algorithm (in a worst-case scenario). 
This is done by first expressing the hitting time as the expectation value of the inverse of an operator $H$, which is obtained by a simple 
transformation of the Markov chain stochastic matrix. We then used results from~\cite{CKS15} to apply $1/H$; in this particular case,  $H$ is positive and we showed that the implementation of $1/H$ can be done more efficiently than the algorithm in~\cite{CKS15}, in terms of the condition number of $H$. 
Such an expected value can be computed using methods for amplitude estimation. 
%The algorithm has resources requirements that scale quadratically better than the classical case in their dependence on a parameter $\Delta$ which relates to the hitting time.
For constant confidence level ($c \approx 0.81$), 
the use of amplitude estimation limits us to a complexity dependence that is $\tilde O(1/\epsilon)$, where $\epsilon$ is the absolute precision of our estimate.  
It is possible to increase the confidence level towards $c$ with an increase in complexity that is $O(\log(|1-c|))$.

%Nevertheless, this too is a quadratic improvement over the classical scaling which is $\tilde{O(1/\epsilon^2)}$. 

\section{Acknowledgements}
AC was supported by a Google research award.
RS acknowledges support of the
LDRD program at LANL and NRO. We thank Sergio Boixo and Ryan Babbush at Google Quantum A.I. Lab Team
for discussions.

%%%%%%%%%%%%%%%%%%%%%%%%%%%%%%%%%%%%%%%%%%%%%%%%%
%\bibliography{GibbsHit}

%%%%%%%%%%%%%%%%%%%%%%%%%%%%%%%%%%%%%%%%%%%%%%%%%

\begin{appendix}
\section{Simulation of $\tilde H$}
\label{appx}
%{\color{red} Anirban, verify all this.}
We provide a method to simulate $\tilde H$ in 
time polylogarithmic in $1/\epsilon$, as required by the algorithm for estimating hitting times
of Sec.~\ref{sec:HT}.
We assume that there exists a procedure $Q_P$ that computes the locations 
and magnitude of the nonzero entries
of the matrix $P$. More precisely, $Q_P$ performs the map
\begin{align}
Q_P \ket \sigma &= \ket \sigma \ket{\sigma'_1,\ldots,\sigma'_d} \otimes
\\
\nonumber
&\otimes \ket{\Pr(\sigma|\sigma'_1),\Pr(\sigma'_1|\sigma),\ldots
,\Pr(\sigma|\sigma'_d),\Pr(\sigma'_d|\sigma) } \;,
\end{align}
where $d$ is the sparsity of $P$, i.e., the largest number of nonzero
matrix elements per row or column. The transition probabilities
are assumed to be exactly represented by a constant number of bits
and we disregard any rounding-off errors.
We also assume access to the oracle $Q_\cU$ such that
\begin{align}
Q_\cU \ket \sigma = - \ket \sigma \ {\rm if} \ \sigma \in \cU
\end{align}
and $Q_\cU \ket \sigma=\ket \sigma$ otherwise.

The Hamiltonians $\bar H$ and $H$ can be constructed as follows. 
For each pair $(\sigma,\sigma')$, such that $\sigma \ne \sigma'$ and $\Pr(\sigma|\sigma') \ne 0$, we define an unnormalized state
\begin{align}
\ket{\mu_{\sigma,\sigma'}}= \frac 1 {\sqrt 2}(\sqrt{\Pr(\sigma|\sigma')} \ket{\sigma'} - \sqrt{\Pr(\sigma'|\sigma)} \ket{\sigma}) \;.
\end{align}
Then, if $\sigma_f \ne \sigma_0$,
\begin{align}
\bra{\sigma_f} \sum_{\sigma,\sigma'} \ket{\mu_{\sigma,\sigma'}}\bra{\mu_{\sigma,\sigma'}} {\sigma_0}\rangle = - \sqrt{\Pr(\sigma_f|\sigma_0)\Pr(\sigma_0|\sigma_f)} \;,
\end{align}
and if $\sigma_f=\sigma_0$,
\begin{align}
\nonumber
\bra{\sigma_0} \sum_{\sigma,\sigma'} \ket{\mu_{\sigma,\sigma'}}\bra{\mu_{\sigma,\sigma'}} {\sigma_0} \rangle&= \sum_{\sigma' \ne \sigma}\Pr(\sigma'|\sigma) \\
& = 1- \Pr(\sigma|\sigma) \;.
\end{align}
These are the same matrix entries of $\bar H$ and the implication is that
\begin{align}
\label{eq:barHsum}
\bar H= \sum_{\sigma,\sigma'} \ket{\mu_{\sigma,\sigma'}}\bra{\mu_{\sigma,\sigma'}} \;.
\end{align}
This is the desired representation of the $\bar H$ as a sum of positive operators. In particular, we can normalize the states
and define
\begin{align}
& \ket{\bar \mu_{\sigma,\sigma'}}=\frac{\ket{\mu_{\sigma,\sigma'}}}{\|\ket{\mu_{\sigma,\sigma'}}\|} \; ,
\\
&\sqrt{\bar \alpha_{\sigma,\sigma'}}=\|\ket{\mu_{\sigma,\sigma'}}\| =\sqrt{\frac {\Pr(\sigma|\sigma')+ \Pr(\sigma'|\sigma)}2}  \;.
\end{align}
Then,
\begin{align}
\bar H = \sum_{\sigma,\sigma'} \bar \alpha_{\sigma,\sigma'} \ket{\bar \mu_{\sigma,\sigma'}} \bra{\bar \mu_{\sigma,\sigma'}} \;.
\end{align}
We let $\Pi_\cU$ be the projector into the subspace $\cU$.
The Hamiltonian is $H= \Pi_\cU \bar H \Pi_\cU$, and using Eq.~\eqref{eq:barHsum}, we obtain
\begin{align}
\nonumber
 H = &  \sum_{\sigma,\sigma' \in \cU} \bar \alpha_{\sigma,\sigma'} \ket{\bar \mu_{\sigma,\sigma'}} \! \bra{\bar \mu_{\sigma,\sigma'}} + \\
 \label{eq:barHproj}
 &+ \! \! \sum_{\sigma' \in \cU} \left(\sum_{\sigma \in \cM}
\! \! \! \Pr(\sigma|\sigma') \right)\ket{\sigma'}\! \bra {\sigma'} ,
\end{align}
which is the desired decomposition as a linear combination of rank-1 projectors.

To build $\tilde H$, we need to take square roots of the projectors. In principle, the dimension  $N_\cU$ is large and we 
want to avoid a presentation of $\tilde H$ as a sum of polynomially many terms in $N_\cU$.
We are also interested in a decomposition of $\tilde H$ in terms of simple unitary operations
so that we can use the results of~\cite{BCC+15} to devise a method to simulate $\exp(-i \tilde H t)$.
We begin with the second term in the right hand side of Eq.~\eqref{eq:barHproj}. Its square root is
\begin{align}
\sum_{\sigma' \in \cU} \left( \sqrt{\sum_{\sigma \in \cM} \Pr(\sigma|\sigma')} \right) \ket {\sigma'} \! \bra{\sigma'} \; .
\end{align}
This term can be simply obtained as a sum of two diagonal unitary operations:
\begin{align}
\frac 1 2 (U_D + U_D^\dagger) \;.
\end{align}
$U_D$ applies a phase to the state $\ket{\sigma'}$ as 
\begin{align}
U_D \ket{ \sigma ' }= e^{i \theta _{\sigma'}} \ket {\sigma '}
\end{align}
with
\begin{align}
\cos (\theta_{\sigma'}) = \sqrt{\sum_{\sigma \in \cM} \Pr(\sigma|\sigma')} \; ,
\end{align}
if $\sigma' \in \cU$. Otherwise, $U_D \ket {\sigma'} = i\ket {\sigma'} $.
$U_D$  can then be implemented by first using $Q_\cU$
to detect if $\sigma'$ is in $\cU$ or not. It next applies $Q_P$
and computes $\theta_{\sigma'}$ in an additional register. Conditional
on the value of $\theta_{\sigma'}$, it applies the corresponding phase
to $\ket{\sigma'}$. It then applies the inverse of $Q_P$ to undo the computation.
That is, $U_D$ requires $O(1)$ uses of $Q_\cU$ and $Q_P$,
and the additional gate complexity is $ O(d)$ due to the computation of $\theta_{\sigma'}$.

The first term in the right hand side of Eq.~\eqref{eq:barHproj}
can be written as a sum of $K'=O(d^2)$ terms as follows.
Using $Q_P$ we can implement a coloring of the graph $G$ with vertex set $V(G) = \{ \sigma : \sigma \in \cU\}$
and edge set $E(G)=\{ (\sigma,\sigma'): \sigma,\sigma' \in U , \Pr(\sigma|\sigma')\ne 0\}$.
We can use the same coloring as that described in~\cite{BCC+14}, which uses a bipartite graph coloring
and was used for Hamiltonian simulation.
Each of the $K'$ terms corresponds to one color and is then a sum of commuting rank-1 projectors.
That is, the first term in the right hand side of Eq.~\eqref{eq:barHproj} is $ \sum_{k=1}^{K'} h_k$ and
\begin{align}
\label{eq:coloring1}
 h_k = \sum_{\sigma,\sigma' \in c_k} \bar \alpha_{\sigma,\sigma'} \ket{\bar \mu_{\sigma,\sigma'}} \bra{\bar \mu_{\sigma,\sigma'}}\;,
\end{align}
where $c_k$ are those elements of $E(G)$ associated with  the $k$-th color.
By the definition of coloring, each rank-1 projector in Eq.~\eqref{eq:coloring1} is orthogonal and commutes
with each other, and then
\begin{align}
\sqrt{ h_k }= \sum_{\sigma,\sigma' \in c_k}\sqrt{ \bar \alpha_{\sigma,\sigma'} }\ket{\bar \mu_{\sigma,\sigma'}} \bra{\bar \mu_{\sigma,\sigma'}} \;.
\end{align}
We can write
\begin{align}
\sqrt{ h_k} = \frac{-iZ_k + iZ_k^\dagger}{2} \;,
\end{align}
where $Z_k$ is the unitary
\begin{align}
Z_k= \exp \left( i \sum_{\sigma,\sigma' \in c_k} \delta_{\sigma,\sigma'} \ket{\bar \mu_{\sigma,\sigma'}} \bra{\bar \mu_{\sigma,\sigma'}}  \right) \;.
\end{align}
The coefficients are chosen so that
\begin{align}
\sin( \delta_{\sigma,\sigma'} )= \sqrt{ \bar \alpha_{\sigma,\sigma'} } \;,
\end{align}
and $0 \le \bar \alpha_{\sigma,\sigma'} \le 1$.

We can simulate each $Z_k$ as follows. Note that
\begin{align}
\label{eq:Zk}
Z_k \ket \sigma =  \xi_{\sigma,\sigma'} \ket \sigma +\xi'_{\sigma,\sigma'} \ket{\sigma'}
\end{align}
where $\sigma'$ is such that $(\sigma,\sigma') \in c_k$. The complex coefficients 
$ \xi_{\sigma,\sigma'}$ and $ \xi'_{\sigma,\sigma'}$ can be simply obtained from the
$\delta_{\sigma,\sigma'}$, and depend only on $\Pr(\sigma|\sigma')$ and  $\Pr(\sigma'|\sigma)$.
Then, on input $\ket \sigma$, we first use $Q_\cU$ to decide whether $\ket \sigma \in \cU$  or not.
 We then apply $Q_P$
once and look for $\sigma '$ such that $(\sigma,\sigma') \in c_k$. We use an additional register
to write a classical description of a quantum circuit that implements the transformation
in Eq.~\eqref{eq:Zk}. 
%
%and a register that holds  $\xi_{\sigma,\sigma'}$ and  $\xi'_{\sigma,\sigma'}$.
We apply the inverse of $Q_P$ and $Q_\cU$ and only keep the last register.
This is sufficient information to apply the map in Eq.~\eqref{eq:Zk}.
%which we implement
%conditional on the values of $\xi_{\sigma,\sigma'}$ and  $\xi'_{\sigma,\sigma'}$.
We can then erase all the additional registers by applying the inverse of the operation
that computed  the quantum circuit.
This works because the quantum circuit is invariant under the permutation of $\sigma$ and $\sigma'$.
To implement $Z_k$ we need to use $Q_\cU$ and $Q_P$, $O(1)$ times. The additional gate complexity is $ O(d \log (N))$
for searching for $\sigma'$ and describing the quantum circuit.

In summary, we found a decomposition of $\tilde H$ as
\begin{align}
\tilde H =  \frac 1 2 \sum_{k=1}^{\tilde K}  \tilde U_k
\end{align}
where $\tilde U_k$ are unitaries. The number of terms is $\tilde K = O(d^2)$.
Using Eq.~\eqref{eq:unitdecomp} and the results above, each $\tilde U_k$ can be implemented with $O(1)$ uses of $Q_\cU$ and $Q_P$,
and at most $ O(d \log (N))$ additional gates. 
%One of the $U_k$ requires $\tilde O(d)$ gates to be implemented. 
%The coefficients $\tilde \alpha_k$ are simple to calculate with $O(1)$ complexity
%and can be made uniform using the decompositions described above.

Using the results of~\cite{BCC+15} (see Sec.~\ref{sec:maintools}), the complexity for simulating $\exp(-i \tilde H t)$ within precision $\epsilon$
for this case is as follows. The number of uses of $Q_\cU$ and $Q_P$ is
\begin{align}
O \left( \tau \log(\tau/\epsilon) / \log\log(\tau/\epsilon)\right) \;,
\end{align}
where $\tau = |t| d^2$.
The additional gate complexity is
\begin{align}
O \left( d \log (N) \tau \log(\tau/\epsilon) / \log\log(\tau/\epsilon)\right) \;.
\end{align}
If we write $C_\cU$ and $C_P$ for the gate complexities of $Q_\cU$ and $Q_P$, respectively, the overall gate complexity to simulate the 
evolution under $\tilde H$ is
\begin{align}
\label{eq:finalWcomp}
C_W(t,\epsilon) = O \left( (d \log (N) +C_\cU+C_P)\tau \frac{\log(\tau/\epsilon) }{ \log\log(\tau/\epsilon)}\right) \;.
\end{align}

\end{appendix}

\end{document}